\documentclass[letterpaper, 10 pt, conference]{ieeeconf}

\IEEEoverridecommandlockouts  
\usepackage{stmaryrd}
\usepackage{soul,xcolor}
\setstcolor{red}
\usepackage{url}
\usepackage{cite}
\usepackage{amsmath,amssymb,amsfonts}
\usepackage{algorithmic}
\usepackage{textcomp}
\usepackage{geometry}
\usepackage{graphicx}
\usepackage{color}
\def\defas{\ensuremath{\mathrel{:=}}}

\DeclareMathOperator{\id}{id}

\def\Set#1#2{\ensuremath{
\left\{#1\,\middle|\,#2\right\}
}}

\def\norm#1{\|  #1 \| }

\def\abs#1{|  #1 | }

\DeclareMathDelimiter{\orbrack}{\mathopen}{operators}{"5D}{largesymbols}{"03}
\DeclareMathDelimiter{\clbrack}{\mathclose}{operators}{"5B}{largesymbols}{"02}
\def\intcc#1{\ensuremath{[#1]}}
\def\intoo#1{\ensuremath{\orbrack#1\clbrack}}
\def\intoc#1{\ensuremath{\orbrack#1]}}
\def\intco#1{\ensuremath{[#1\clbrack}}

\def\Hintcc#1{\ensuremath{\llbracket#1\rrbracket}}

\newtheorem{theorem}{Theorem}
\newtheorem{lemma}[theorem]{Lemma}
\newtheorem{prop}[theorem]{Proposition}

\newtheorem{remark}{Remark}

\newtheorem{assumption}{Assumption}

\newtheorem{corollary}{Corollary}

\title{ \LARGE \bf Safe Domains of Attraction for Discrete-Time  Nonlinear Systems: Characterization and Verifiable Neural Network Estimation }
\author{Mohamed Serry*, Haoyu Li*, Ruikun Zhou*, Huan Zhang, and Jun Liu 
\thanks {*Equal contribution}
\thanks{Mohamed Serry, Ruikun Zhou, and  Jun Liu are with the Department of Applied Mathematics, University of Waterloo, Waterloo, Ontario, Canada.  Email: \texttt{\{mserry, ruikun.zhou, j.liu\}@uwaterloo.ca}} 
\thanks{Haoyu Li is with the Department of Computer Science and Huan Zhang is with the Department of Electrical and Computer Engineering, University of Illinois Urbana-Champaign,
Urbana, IL 61801, USA. Email: \texttt{haoyuli5@illinois.edu, huan@huan-zhang.com}}
\thanks{This work was funded in part by the NSERC Discovery Grant, the Canada Research Chairs program, the U.S. National Science Foundation (NSF) IIS-2331967, and the Schmidt Science AI2050 program.}  }

\begin{document}
\maketitle
\begin{abstract}
Analysis of nonlinear autonomous systems typically involves estimating domains of attraction, which have been a topic of extensive research interest for decades.   Despite that, accurately estimating domains of attraction for nonlinear systems remains a challenging task, where existing methods are conservative or limited to low-dimensional systems. The estimation becomes even more challenging when accounting for state constraints. In this work, we propose a framework to accurately estimate safe (state-constrained) domains of attraction for discrete-time autonomous nonlinear systems. In establishing this framework, we first derive a new  Zubov equation, whose solution corresponds to the exact safe domain of attraction. The solution to the aforementioned Zubov equation is shown to be unique and continuous over the whole state space. We then present a physics-informed approach to approximating the solution of the Zubov equation using neural networks. To obtain certifiable estimates of the domain of attraction from the neural network approximate solutions,  we propose a verification framework that can be implemented using standard verification tools (e.g., $\alpha,\!\beta$-CROWN and dReal). To illustrate its effectiveness, we demonstrate our approach through numerical examples concerning nonlinear systems with state constraints.
\end{abstract}

\begin{keywords}
Safety, formal verification, neural networks, nonlinear systems, stability analysis, Zubov's theorem
\end{keywords}

\section{Introduction}
 The safe (i.e., state-constrained) domain of attraction (DOA) of a given dynamical system is the set of state values from which trajectories are guaranteed to converge to an equilibrium point of interest under the system's dynamics, while satisfying specified state constraints. Such a set provides a safe operation region, while  ensuring attractivity to the equilibrium point. DOAs are very prevalent, especially in safe stabilization scenarios, and that has motivated an immense amount of research work on computing or approximating DOAs. In this paper, we consider the problem of estimating the state-constrained DOA of a general discrete-time autonomous nonlinear system. 

In the literature, DOAs are predominantly estimated using the framework of Lyapunov functions, where candidate Lyapunov functions of fixed templates (e.g., quadratic forms and sum-of-squares polynomials \cite{packard2010help}) are typically assumed. Then, the parameters of such templates are tuned to satisfy the standard Lyapunov conditions, or the  more relaxed multi-step  and  non-monotonic Lyapunov conditions \cite{ahmadi2008non,bobiti2014computation}.  Lyapunov-based approaches utilizing fixed templates are generally restrictive, providing, if existent, conservative estimates of DOAs \cite{edwards2024fossil}.   

Interestingly, if an initial certifiable DOA estimate is provided (e.g., using quadratic Lyapunov functions),  DOAs can be underapproximated arbitrarily using  iterative computations of the backward reachable set of the initial DOA estimate, where such iterations are guaranteed to converge to the exact DOA \cite{balint2006methods, serry2024safe}. However, the complexity, in terms of the set representation of each DOA estimate, increases  with each iteration, making the resulting estimates impractical in formal verification tasks.

Recently, there has been a growing interest in using learning-based approaches to estimate DOAs, where neural networks are trained to satisfy standard Lyapunov conditions  and then verification tools (e.g., interval arithmetic and mixed-integer programming) are implemented to ensure that the trained neural networks provide certifiable DOA estimates \cite{shi2024certified, yang2024lyapunov,wu2023neural,dai2021lyapunov,chen2021learning, chang2019neural}.  Despite the high computational efficiency associated with training neural networks, neural network verification typically suffers from high computational demands due to state-space discretization. Additionally,  the resulting DOA estimates do not significantly outperform the standard Lyapunov-based approaches using fixed templates. Interestingly, there have been promising developments in the field of neural-network verification, where computationally efficient linear bound propagation  and branch and bound have been utilized, enabling fast and scalable neural network  verification  \cite{bunel2020branch,xu2021fast, wang2021beta,ferrari2022complete,zhou2024scalable}.

  For some classes of  nonlinear systems with local exponential stability properties, DOAs can be  characterized  as  sublevel sets of particular value functions, which are solutions to   functional-type equations: the maximal Lyapunov and  Zubov equations \cite{balint2006methods,giesl2007determination,oshea1964extension,xue2020characterization}. Zubov equations are preferable when estimating DOAs as their solutions are bounded, where these solutions have been typically estimated numerically using discretization-based  approximations \cite{xue2020characterization} and sum-of-squares optimization \cite{xue2020robust}, which are limited to low-dimensional systems. Still, the Zubov-based approaches are advantageous in the sense that, in theory, accurately approximating the solutions to the Zubov equation  provides  large  DOA estimates. 

In this work, motivated by the utilities of neural-network approximations, the advancements in neural network verification, and the theoretical advantages associated with Zubov-based methods, we propose a DOA estimation framework for discrete-time autonomous nonlinear systems  that relies on neural network approximations of solutions to a  new Zubov equation that accounts for state constraints. 

Zubov equations have been developed for discrete-time nonlinear systems without \cite{oshea1964extension} and with \cite{xue2020characterization} state constraints. Interestingly, the framework in \cite{xue2020characterization} even accounts for disturbances. The Zubov equation in \cite{xue2020characterization} contains a non-smooth term  (in terms of the $\min$ function) to account for state constraints. This non-smooth term can make neural network training, which typically relies on gradient-based optimization methods, more challenging. In this work, and by tailoring value functions that correspond to the safe DOA, we present a new Zubov equation that does not possess this non-smooth term, making it more suited for neural network training. In addition, the framework in \cite{xue2020characterization} assumes boundedness of the safe domain in addition to strong Lipschitz-type conditions   on the system's dynamics and the state constraints, which we relax in our framework.

Mere neural network approximate solutions to Zubov equations do not provide certifiable DOA estimates, as neural network training does not account for approximation errors. To obtain certifiable estimates from the neural-network approximations,  we propose a verification framework, which is a discrete-time variation  of the verification framework proposed in \cite{liu2025physics}, where certifiable ellipsoidal DOA estimates and backward reachability computations are employed. This verification framework can be implemented using standard verification tools such as $\alpha,\beta$-CROWN \cite{zhang2018efficient, xu2020automatic, xu2021fast, wang2021beta, zhou2024scalable, shi2024genbab} and dReal \cite{gao2013dreal}.

The organization of this paper is as follows. The necessary preliminaries and notation are introduced in Section \ref{sec:Preliminaries}. The problem setup is introduced in Section \ref{sec:ProblemSetup}. Some properties of the safe DOA are discussed in Section \ref{sec:DOAProperties}. The new value functions are presented in Section \ref{sec:ValueFunctions}. The DOA  characterization in terms of these  value functions is introduced in Section \ref{sec:DOACharcterization}. The properties of the value functions are presented in Section \ref{sec:VFProperties}.  The Zubov and Lyapunov functions corresponding to the value functions are discussed in Section \ref{sec:Zubov}. The neural network approximation is presented in Section \ref{sec:PINN}. The verification framework is presented in Section \ref{sec:NNVerification}. The proposed method is illustrated  through three numerical examples in Section \ref{sec:NumericalExamples}. Finally, the study is concluded in Section \ref{sec:Conclusion}.

\section{Notation and Preliminaries}
\label{sec:Preliminaries}
Let $\mathbb{R}$, $\mathbb{R}_+$,  $\mathbb{Z}$, and $\mathbb{Z}_{+}$ denote
the sets of real numbers, non-negative real numbers, integers, and
non-negative integers, respectively, and
$\mathbb{N} = \mathbb{Z}_{+} \setminus \{ 0 \}$.
Let $\intcc{a,b}$, $\intoo{a,b}$,
$\intco{a,b}$, and $\intoc{a,b}$
denote closed, open and half-open
intervals, respectively, with end points $a$ and $b$, and
 $\intcc{a;b}$, $\intoo{a;b}$,
$\intco{a;b}$, and $\intoc{a;b}$ stand for their discrete counterparts,
e.g.,~$\intcc{a;b} = \intcc{a,b} \cap \mathbb{Z}$, and
$\intco{1;4} = \{ 1,2,3 \}$.  In $\mathbb{R}^{n}
$, the relations $<$, $\leq$, $\geq$, and
$>$ are defined component-wise, e.g., $a < b$, where $a,b\in \mathbb{R}^{n}$, iff $a_i < b_i$ for
all $i\in  \intcc{1;n}$. For $a, b \in \mathbb{R}^n$, $a \leq b$,
the closed hyper-interval (or hyper-rectangle) $\Hintcc{a,b}$ denotes the set $\Set{x\in \mathbb{R}^{n}}{a\leq x\leq b}$. Let $\norm{\cdot}$ and $\norm{\cdot}_{\infty}$ denote the Euclidean and maximal norms on $\mathbb{R}^{n}$, respectively, and $\mathbb{B}_{n}$ be the $n$-dimensional closed unit ball induced by $\norm{\cdot}$. The $n$-dimensional zero vector is denoted by $0_{n}$.  Let $\id_{n}$ denote the $n\times n$ identity matrix. For $A\in \mathbb{R}^{n\times m}$, $\norm{A}$ and  $\norm{A}_{\infty}$ denote the matrix norms of $A$ induced by the Euclidean and maximal norms, respectively. Given $x\in \mathbb{R}^{n}$  and $A\in \mathbb{R}^{n\times m}$, $\abs{x}\in \mathbb{R}^{n}_{+}$ and $\abs{A}\in \mathbb{R}_{+}^{n\times m}$ are defined as $|x|_{i}\defas |x_{i}|,~i\in \intcc{1;n}$, and $\abs{A}_{i,j}\defas \abs{A_{i,j}},~(i,j)\in \intcc{1;n}\times \intcc{1;m}$, respectively. Let $\mathcal{S}^{n}$ denote the set of $n\times n$ real symmetric matrices. Given $A\in \mathcal{S}^{n}$, $\underline{\lambda}(A)$ and $\overline{\lambda}(A)$ denote the minimum and maximum eigenvalues of $A$, respectively.  Let $\mathcal{S}_{++}^{n}$ denote the set of $n\times n$ real symmetric positive definite matrices $\Set{A\in \mathcal{S}^{n}}{\underline{\lambda}(A)>0}$. Given $A\in \mathcal{S}_{++}^{n}$, $A^{\frac{1}{2}}$ denotes the unique real symmetric positive definite matrix $K$ satisfying $A=K^2$ \cite[p.~220]{Abadir2005matrix}.  The interior and the boundary  of $X\subseteq \mathbb{R}^{n}$ are denoted by $\mathrm{int}(X)$ and $\partial X$, respectively. Given $f\colon X\rightarrow Y$ and  $P\subseteq X$, the image  of $f$ on $P$ is defined as  $f(P)\defas\Set{f(x)}{x\in P}$.     Given $f\colon X \rightarrow X$ and  $x\in X$,  $f^{0}(x)\defas x$, and for  $M\in \mathbb{N}$, we define $f^{M}(x)$ recursively as follows:   $f^{k}(x)=f(f^{k-1}(x)),~k\in \intcc{1;M}$.   A subset $S\subseteq X$ is said to be  invariant under a mapping $f\colon X\rightarrow X$ if $f(X)\subseteq X$.

\section{Problem Setup}
\label{sec:ProblemSetup}
Consider the discrete-time system
\begin{equation}\label{eq:System}
x_{k+1}=f(x_{k}),~k\in \mathbb{Z}_{+},
\end{equation}
where $x_{k}\in \mathbb{R}^{n}$ is the state and $f\colon \mathbb{R}^{n}\rightarrow \mathbb{R}^{n}$ is the system's transition function. The trajectory of system \eqref{eq:System} starting from $x_{0}\in \mathbb{R}^{n}$ is the function $\varphi_{x}\colon\mathbb{Z}_{+}\rightarrow \mathbb{R}^{n}$, satisfying:
\begin{align*}
    \varphi_{x}(0)&=x,\\
    \varphi_{x}(k+1)&=f(\varphi_{x}(k))=f^{k+1}(x),~k\in \mathbb{Z}_{+}.
\end{align*}

\begin{assumption}\label{Assumptions}
$f$ is continuous over $\mathbb{R}^{n}$,  $0_{n}$ is an equilibrium point of system \eqref{eq:System} (i.e., $f(0_{n})=0_{n}$), and
    $0_{n}$ is locally exponentially stable. That is, there exist fixed parameters $r\in \intoo{0,\infty}$,  $M\in \intco{1,\infty}$, and $\lambda \in \intoo{0,1}$ such that for all $x\in r\mathbb{B}_{n}$,
$
\norm{\varphi_{x}(k)}\leq M \lambda^{k}\norm{x},~k\in \mathbb{Z}_{+}. 
$   
\end{assumption}
Let $\mathcal{X}\subseteq \mathbb{R}^{n}$ be a safe set.  We make the following assumption. 
\begin{assumption}
$\mathcal{X}$ is open and $0_{n}\in \mathcal{X}$.
\end{assumption}

Define the safe DOA inside $\mathcal{X}$ as 
\begin{equation*}
\mathcal{D}_{0}^{\mathcal{X}}\defas\Set{x\in \mathcal{X}}{\varphi_{x}(k)\in \mathcal{X},\, \forall k\in \mathbb{Z}_{+}, \,\lim_{k\rightarrow \infty} \varphi_{x}(k)=0_{n}}.
\end{equation*}

Any invariant subset of $\mathcal{D}_{0}^{\mathcal{X}}$ under $f$, containing $0_{n}$ in its interior  is called a safe region of attraction (ROA) in $\mathcal{X}$. Our goal is to compute a large safe ROA that  closely approximates  $\mathcal{D}_{0}^{\mathcal{X}}$.

\section{Properties of the DOA}
\label{sec:DOAProperties}
In this section, we introduce some important properties of the safe DOA $\mathcal{D}_{0}^{\mathcal{X}}$, which will be utilized in the proofs of the main results of this work.

\begin{theorem}
    The set $\mathcal{D}_{0}^{\mathcal{X}}\subseteq \mathcal{X}$ is nonempty, invariant under $f$, and open.
\end{theorem}
\begin{proof}
   The non-emptiness follows from the fact that $0_{n}\in\mathcal{D}_{0}^{\mathcal{X}}$. The invariance property can be deduced as follows: for  $x\in \mathcal{D}_{0}^{\mathcal{X}}$, $\varphi_{x}(k)\in \mathcal{X}~\forall k\in \mathbb{Z}_{+}$ and $\lim_{k\rightarrow \infty}\varphi_{x}(k)=0_{n}$. This implies that  $\varphi_{x}(k+1)=\phi_{f(x)}(k)\in \mathcal{X}~\forall k\in \mathbb{Z}_{+}$ and $\lim_{k\rightarrow \infty}\varphi_{x}(k+1)=\lim_{k\rightarrow \infty}\varphi_{f(x)}(k)=0_{n}$. Hence, $f(x)\in \mathcal{D}_{0}^{\mathcal{X}}$.
   
Now, we prove that $\mathcal{D}_{0}^{\mathcal{X}}$ is open.  Recall the definitions of $M,~r,~\lambda$ in Assumption \ref{Assumptions}. Let $\theta\in \intoo{0,\infty}$ be such that 
    $
\theta \mathbb{B}_{n}\subseteq \mathcal{X},
    $
    which exists due to the openness of $\mathcal{X}$ and the fact that $0_{n}\in \mathcal{X}$. Fix $x_{0}\in \mathcal{D}_{0}^{\mathcal{X}}$, then due to the convergence of $\varphi_{x_{0}}$ to $0_{n}$ within $\mathcal{X}$, there exists $N\in \mathbb{Z}_{+}$ such that  
    $
\varphi_{x_{0}}(j)\in \mathcal{X}~\forall j\in \intcc{0;N-1}
    $
and 
    $
\varphi_{x_{0}}(N)\in \frac{\tilde{r}}{2}\mathbb{B},
    $
where  $\tilde{r}$ satisfies
$
0<\tilde{r}\leq \min \{\frac{\theta}{M}, r\}.
$
Let $\rho_{j},~j\in \intcc{0;N-1}$, be positive numbers satisfying:
$
\varphi_{x_{0}}(j)+\rho_{j}\mathbb{B}_{n}\subset \mathcal{X},~j\in \intcc{0;N-1}.
$
Such numbers exist due to the openness of $\mathcal{X}$.
Note that 
$\varphi_{y}(j)=f^{j}(y)$
for all $y\in \mathbb{R}^{n}$ and $j\in \mathbb{Z}_{+}$.  As $f$ is continuous at $x_{0}$, $f^{2}$, ...  ,$f^{N}$ are also continuous at $x_{0}$. Therefore, there exists $\delta\in \intoo{0,\infty}$ such that, for all $x\in x_{0}+\delta \mathbb{B}_{n}$, 
$
\norm{f^{j}(x)-f^{j}(x_{0})} < \rho_{j},~j\in \intcc{0;N-1}$, and 
$\norm{f^{N}(x)-f^{N}(x_{0})}<\frac{\tilde{r}}{2}$. Consequently, we have for all $x\in x_{0}+\delta \mathbb{B}_{n}$,
$
\varphi_{x}(j)=f^{j}(x)\in \varphi_{x_{0}}(j)+\rho_{j}\mathbb{B}_{n}\subset \mathcal{X},~j\in \intcc{0;N-1},
$
and
$\norm{\varphi_{x}(N)}=\norm{f^{N}(x)}
\leq \norm{f^{N}(x)-f^{N}(x_{0})}+\norm{f^{N}(x_{0})} \leq \tilde{r}\leq r$.  The local exponential stability indicates that, for all $x\in x_{0}+\delta \mathbb{B}_{n}$, 
    $
\norm{f^{N+k}(x)}\leq M  \norm{f^{N}(x)} \lambda^{k}\leq M\tilde{r}\lambda^{k} ,~k\in \mathbb{Z}_{+}.
    $
Hence, {\color{red}$\varphi_{x}(j)=f^{j}(x)\rightarrow 0_{n}$} as $j\rightarrow \infty$ for all $x\in x_{0}+\delta \mathbb{B}_{n}$.
    Also, the local exponential stability and the definition of $\tilde{r}$ imply that, for all $x\in x_{0}+\delta \mathbb{B}_{n}$,
$\norm{f^{N+k}(x)}\leq  M \tilde{r} \leq \theta,~k\in \mathbb{Z}_{+}
\Rightarrow \varphi_{x}(N+k)=f^{N+k}(x)\in \theta \mathbb{B}_{n}\subseteq \mathcal{X}~\forall k\in \mathbb{Z}_{+}$.
    Therefore, $x\in\mathcal{D}_{0}^{\mathcal{X}}$ for all $x\in x_{0}+\delta \mathbb{B}_{n}$. As $x_{0}\in \mathcal{D}_{0}^{\mathcal{X}}$ is arbitrary, the proof is complete.
\end{proof}
\section{Value Functions}
In this section, we introduce the value functions that can be used to characterize $\mathcal{D}_{0}^{\mathcal{X}}$ and to derive Lyapunov and Zubov type equations.
\label{sec:ValueFunctions}
To this end, we
    let $\alpha\colon \mathbb{R}^{n}\rightarrow \mathbb{R}_{+}$ be a positive definite continuous function such that 
$$
\alpha_{m}\norm{x}^{2}\leq \alpha(x)\leq \alpha_{M}\norm{x}^{2},~x\in \mathbb{R}^{n},
$$
for some $\alpha_{m}, \alpha_{M}\in \intoo{0,\infty}$.  The definition of $\alpha$ and the parameters $\alpha_{m}$ and $\alpha_{M}$ are fixed throughout the following discussion. 
\begin{assumption}\label{assump:gamma}
    There exists a function $\gamma \colon \mathbb{R}^{n} \rightarrow \mathbb{R}_{+}\cup\{\infty\}$ satisfying:
    \begin{enumerate}
        \item $\gamma$ is finite and  continuous over $\mathcal{X}$, and there exists $\underline{\gamma}\in \mathbb{R}_{+}\setminus \{0\}$ such that 
        $$
\gamma(x)\geq \underline{\gamma}~\forall x\in \mathcal{X},
        $$
        \item $\gamma(x)=\infty$ whenever $x\notin \mathcal{X}$,
        \item for any sequence $\{x_{n}\}$, with $x_{n}\rightarrow x\in \partial \mathcal{X}$,  $\gamma(x_{n})\rightarrow \infty$.
    \end{enumerate}
\end{assumption}

\begin{remark} \label{rem:transformation}
    If $\mathcal{X}$ is a strict 1-sublevel set of a continuous function $g_{\mathcal{X}}\colon \mathbb{R}^{n}\rightarrow \mathbb{R}$, i.e., $\mathcal{X}=\{x\in \mathbb{R}^{n}|g_{\mathcal{X}}(x)<1\}$, then we can define 
    $$
\gamma(x)=1+\frac{1}{\mathrm{ReLu}(1-g_{\mathcal{X}}(x))},
    $$
    where $1/0\defas \infty$ and $\mathrm{ReLu}\colon \mathbb{R}\rightarrow \mathbb{R}$ is the rectifier linear unit function defined as $\mathrm{ReLu}(x)=(x+|x|)/2,~x\in \mathbb{R}$. With this definition of $\gamma$, the conditions of Assumption \ref{assump:gamma} hold with $\underline{\gamma}=1$. 
\end{remark}

We define the  value functions $\mathcal{V}\colon \mathbb{R}^{n}\rightarrow \mathbb{R}_{+}\cup \{\infty\}$ 
and $\mathcal{W}\colon \mathbb{R}^{n}\rightarrow \intcc{0,1}$ 
as follows: 

\begin{equation}\label{eq:Lyapunov1}
\mathcal{V}(x)\defas \sum_{k=0}^{\infty} \gamma(\varphi_{x}(k))\alpha(\varphi_{x}(k)),
\end{equation}
and
\begin{equation}\label{eq:Lyapunov2}
  \mathcal{W}(x)\defas 1-\exp(-\mathcal{V}(x)),
\end{equation}
where $\exp(-\infty)\defas 0$.

\section{Characterizing the DOA Using the Value Functions}
In this section, we charcterize the safe DOA in terms of the sublevel sets corresponding tho the value functions $\mathcal{V}$ and $\mathcal{W}$.
\label{sec:DOACharcterization}

\begin{theorem}
\begin{align*}
\mathcal{D}_{0}^{\mathcal{X}}&=
\mathbb{V}_{\infty}\defas \Set{x\in \mathbb{R}^{n}}{\mathcal{V}(x)<\infty}\\
&=\mathbb{W}_{1}\defas \Set{x\in \mathbb{R}^{n}}{\mathcal{W}(x)<1}.
\end{align*}
\end{theorem}

\begin{proof}
   Due to the one-to-one correspondence between the codomains of $\mathcal{V}$ 
 and $\mathcal{W}$ using equation \eqref{eq:Lyapunov2}, it is sufficient to only show that $\mathbb{V}_{\infty}=\mathcal{D}_{0}^{\mathcal{X}}$. Let $x\in \mathbb{V}_{\infty}$. We will show that ${\varphi}_{x}(k)\in \mathcal{X}$ for all $k\in \mathbb{Z}_{+}$.  By contradiction, assume that ${\varphi}_{x}(N)\notin \mathcal{X}$ for some $N\in \mathbb{Z}_{+}$. Then by the definition of $\gamma$, 
    $
\gamma({\varphi}_{x}(N))=\infty
    $
    implying that $\mathcal{V}(x)=\infty$, and that contradicts the fact that $x\in \mathbb{V}_{\infty}$. Hence, we have $\varphi_{x}(k)\in \mathcal{X}$ for all $k\in \mathbb{Z}_{+}$. Using the lower bound on $\gamma$, we have 
$
\alpha({\varphi}_{x}(k))\leq \gamma({\varphi}_{x}(k)) \alpha({\varphi}_{x}(k))/\underline{\gamma},~k\in \mathbb{Z}_{+}.
$
As $\sum_{k=0}^{\infty}\gamma({\varphi}_{x}(k)) \alpha({\varphi}_{x}(k))$ is convergent,
then by the  comparison test, {\color{red}$\sum_{k=0}^{\infty} \alpha({\varphi}_{x}(k))$} is also convergent, implying  $\lim_{k\rightarrow \infty} \alpha({\varphi}_{x}(k))=0$. Using the lower bound on $\alpha$, we have
    $
\lim_{k\rightarrow \infty} \norm{\varphi_{x}(k)}^{2}\leq \lim_{k\rightarrow \infty} \frac{1}{\alpha_{m}}\alpha(\varphi_{x}(k))=0.
    $
    Therefore, $x\in\mathcal{D}_{0}^{\mathcal{X}}$.

Now, let $x\in\mathcal{D}_{0}^{\mathcal{X}}$ and  recall the definitions of $M,~r,~\lambda$ in Assumption \ref{Assumptions}. Note that $0<\underline{\gamma}\leq \gamma(\varphi_{x}(j))<\infty$ for all $j\in \mathbb{Z}_{+}$ as $\varphi_{x}(j)\in \mathcal{X},~j\in \mathbb{Z}_{+}$.   Let $\theta\in \intoo{0,\infty}$ be such that 
    $
\theta \mathbb{B}_{n}\subseteq \mathcal{X},
    $
    which exists due to the openness of $\mathcal{X}$ and the fact that $0_{n}\in \mathcal{X}$. Due to the convergence of $\varphi_{x} $ to $0_{n}$, there exists $N\in \mathbb{Z}_{+}$ such that 
$
\varphi_{x}(N)=y\in \tilde{r}\mathbb{B}_{n},
$
where 
$
0<\tilde{r}\leq \min \{\theta/M,r\}.
$
As $\varphi_{x}(N)\in {r}\mathbb{B}_{n}$, the local exponential stability and the definition of $\tilde{r}$ imply that 
$\norm{\varphi_{x}(N+k)}=\norm{\varphi_{y}(k)}\leq M \lambda^{k}\norm{y}
\leq M \lambda^{k}\tilde{r}\leq  \lambda^{k}\theta\leq \theta  ,~k\in \mathbb{Z}_{+}$.
Hence, 
$
\alpha(\varphi_{x}(N+k))\leq \alpha_{M}\norm{\varphi_{x}(N+k)}^{2}\leq \alpha_{M}\lambda^{2k}\theta^{2},~k\in \mathbb{Z}_{+}.
$ Let $\Gamma_{\theta }\in \intoo{0,\infty}$ be such that 
$
\gamma(x)\leq \Gamma_{\theta }~\forall x\in \theta \mathbb{B}_{n},
$
which exists due to the compactness of $\theta \mathbb{B}_{n}$ and the continuity of $\gamma$ over $\theta \mathbb{B}_{n}$. Consequently, we have  
$
\mathcal{V}(x)=\sum_{k=0}^{\infty}\gamma({\varphi}_{x}(k))\alpha({\varphi}_{x}(k))=\sum_{k=0}^{N-1}\gamma({\varphi}_{x}(k))\alpha({\varphi}_{x}(k))+\sum_{k=N}^{\infty}\gamma({\varphi}_{x}(k))\alpha({\varphi}_{x}(k)) \leq \sum_{k=0}^{N-1}\gamma({\varphi}_{x}(k))\alpha({\varphi}_{x}(k))+\alpha_{M}\Gamma_{\theta}\theta^{2}\sum_{k=0}^{\infty}\lambda^{2k} \leq \sum_{k=0}^{N-1}\gamma({\varphi}_{x}(k))\alpha({\varphi}_{x}(k))+\alpha_{M}\Gamma_{\theta}\theta^{2}\frac{1}{1-\lambda^{2}}<\infty.
$   
Hence, $x\in \mathbb{V}_{\infty}$, and that completes the proof.
\end{proof}
\section{Properties of the Value Functions}
\label{sec:VFProperties}
In this section, we state some important properties for the functions $\mathcal{V}$ and $\mathcal{W}$.

\begin{lemma}
    The functions $\mathcal{V}$ and $\mathcal{W}$ are positive definite.
\end{lemma}
This is an immediate consequence of the definitions.
% \begin{proof}
% This follows from the fact that $\mathcal{V}(x)\geq \underline{\gamma}\alpha(x),~x\in \mathbb{R}^{n}$, the positive defniteness of $\alpha$, the fact that $V(0_{n})=0$, and the definition of $\mathcal{W}$ in \eqref{eq:Lyapunov2}.  
% \end{proof}

\begin{theorem} 
    $\mathcal{V}$ is continuous over $\mathcal{D}_{0}^{\mathcal{X}}$.
\end{theorem}

\begin{proof}
     Recall the definitions of $M,~r,~\lambda$ in Assumption \ref{Assumptions}. Let $\theta\in \intoo{0,\infty}$ be such that 
    $
\theta \mathbb{B}_{n}\subseteq \mathcal{X},
    $
    which exists due to the openness of $\mathcal{X}$ and the fact that $0_{n}\in \mathcal{X}$. Let $x_{0}\in\mathcal{D}_{0}^{\mathcal{X}}$ and $\varepsilon>0$ be arbitrary, where we assume without loss of generality that $\varepsilon\leq \min\{\theta/M,r\}$. Then there exists $N\in \mathbb{Z}_{+}$ such that $\varphi_{x_{0}}(N)\in \frac{\varepsilon}{2}\mathbb{B}_{n}$, where the exponential stability indicates that 
     $
\norm{\varphi_{x_{0}}(N+k)}\leq M \varepsilon \lambda^{k}\leq \theta ,~k\in \mathbb{Z}_{+}.
     $
   Let $\Gamma_{\theta }\in \intoo{0,\infty}$ be such that 
$
\gamma(x)\leq \Gamma_{\theta }~\forall x\in \theta \mathbb{B}_{n},
$
which exists due to the compactness of $\theta \mathbb{B}_{n}$ and the continuity of $\gamma$ over $\theta \mathbb{B}_{n}$. Consequently, we have
  $
\sum_{k=N}^{\infty}\gamma (\varphi_{x_{0}}(k)) \alpha (\varphi_{x_{0}}(k))\leq \sum_{k=0}^{\infty} \Gamma_{\theta}\alpha_{M} \norm{\varphi_{x_{0}}(N+k)}^{2}
\leq \Gamma_{\theta}\alpha_{M}\frac{M^{2}\varepsilon^{2}}{1-\lambda^{2}}.
   $
Let $\delta>0$ be such that $x_{0}+\delta \mathbb{B}_{n}\subset\mathcal{D}_{0}^{\mathcal{X}}$ and, for all $x\in x_{0}+\delta \mathbb{B}_{n}$, 
$
    \abs{\sum_{k=0}^{N-1}\gamma(\varphi_{x_{0}}(k))\alpha(\varphi_{x_{0}}(k))- \gamma(\varphi_{x}(k))\alpha(\varphi_{x}(k))}< \varepsilon$, and $
    \norm{\varphi_{x_{0}}(N)-\varphi_{x}(N)}\leq  \frac{\varepsilon}{2}$.
   Such $\delta$ exists due to the openness of $\mathcal{D}_{0}^{\mathcal{X}}$, the continuity of $f$ (and consequently the continuity of $f^{i},~i\in\intcc{1;N}$),  the continuity of $\alpha$ and the continuity of $\gamma$ over $\mathcal{X}$.
Consequently, for all $x\in x_{0}+\delta \mathbb{B}_{n}$, 
$
\norm{\varphi_{x}(N)}\leq \norm{\varphi_{x_{0}}(N)-\varphi_{x}(N)}+\norm{\varphi_{x_{0}}(N)}\leq \varepsilon.
$
The exponential stability indicates that, for all $x\in x_{0}+\delta \mathbb{B}_{n}$, 
   $
\norm{\varphi_{x}(N+k)}\leq M \varepsilon \lambda^{k}\leq \theta ,~k\in \mathbb{Z}_{+}.
   $
   Therefore, 
   $
\sum_{k=N}^{\infty} \gamma (\varphi_{x}(k)) \alpha (\varphi_{x}(k))\leq \sum_{k=0}^{\infty} \Gamma_{\theta}\alpha_{M} \norm{\varphi_{x_{0}}(N+k)}^{2}\leq \Gamma_{\theta}\alpha_{M}\frac{M^{2}\varepsilon^{2}}{1-\lambda^{2}}~\forall x\in x_{0}+\delta \mathbb{B}_{n}$.    Finally we have, for all $x\in x_{0}+\delta \mathbb{B}_{n}$ ($\mathcal{V}(x)$ is well-defined for $x\in x_{0}+\delta \mathbb{B}_{n}$),
 $
       \abs{\mathcal{V}(x_{0})-\mathcal{V}(x)} \leq
 \abs{\sum_{k=0}^{N-1}\gamma(\varphi_{x_{0}}(k))\alpha(\varphi_{x_{0}}(k))-\gamma(\varphi_{x}(k)) \alpha(\varphi_{x}(k))}
       +\sum_{k=N}^{\infty}\gamma (\varphi_{x_{0}}(k)) \alpha (\varphi_{x_{0}}(k))
       +\sum_{k=N}^{\infty} \gamma (\varphi_{x}(k)) \alpha (\varphi_{x}(k))
        \leq \varepsilon+ \Gamma_{\theta}\alpha_{M}\frac{M^{2}\varepsilon^{2}}{1-\lambda^{2}}+\Gamma_{\theta}\alpha_{M}\frac{M^{2}\varepsilon^{2}}{1-\lambda^{2}}$. As $\varepsilon$ is arbitrary, the proof is complete.
\end{proof}
\begin{theorem}  $\mathcal{V}(x_{k})\rightarrow \infty$ whenever $x_{k}\rightarrow x\in \partial\mathcal{D}_{0}^{\mathcal{X}}$,
\end{theorem}
\begin{proof}
    Without loss of generality, consider a sequence $\{x_{k}\}\subseteq \mathcal{D}_{0}^{\mathcal{X}}$, where $x_{k}\rightarrow x\in \partial \mathcal{D}_{0}^{\mathcal{X}}$. Let $\theta \in \intoo{0,\infty}$ be such that
$
\theta<r$ and 
$\theta\mathbb{B}_{n}\subseteq \mathcal{X}.
    $
    Let $\tilde{r}\in \intoo{0, \theta/M^{2}}   $.
    For each $k\in \mathbb{Z}_{+}$, let $T_{k}\in \mathbb{Z}_{+}$ be the first time instance such that 
    $
\varphi_{x_{k}}(T_{k})\in \tilde{r}\mathbb{B}_{n}.
    $
    If the sequence $\{T_{k}\}$ diverges to $\infty$, then we have
    $
\mathcal{V}(x_{k})\geq \sum_{j=0}^{T_{k}-1}\gamma(\varphi_{x_{k}}(j))\alpha(\varphi_{x_{k}}(j))\geq \underline{\gamma}\alpha_{m}\tilde{r}^{2}(T_{k}-1)\rightarrow \infty
    $ as $k\rightarrow \infty$,
    Assume that $\{T_{k}\}$ does not diverge to $\infty$, then there exists a bounded subsequence, again denoted $\{T_{k}\}$, with an upper bound $T\in \mathbb{Z}_{+}$ such that  
    $
T_{k}\leq T,~k\in \mathbb{Z}_{+}.
    $
It follows, as $\tilde{r}\leq \theta/M^{2}\leq \theta<r$, that  
$
\varphi_{x_{k}}(T_{k})\in \tilde{r} \mathbb{B}_{n}\subseteq r \mathbb{B}_{n},~k\in \mathbb{Z}_{+}.
$
Therefore, 
$
\norm{\varphi_{x_{k}}(T_{k}+j)}\leq  M\lambda^{j}\tilde{r}\leq \theta/M,~j,k\in \mathbb{Z}_{+}. 
$
Hence, 
$
\varphi_{x_{k}}(T)\in (\theta/M) \mathbb{B}_{n}\subseteq r\mathbb{B}_{n},~k\in \mathbb{Z}_{+},
$
implying, using the continuity of $f^{T}(\cdot)$, 
$
\varphi_{x}(T)\in (\theta/M) \mathbb{B}_{n}\subseteq r \mathbb{B}_{n}.
$
Therefore, 
$
\norm{\varphi_{x}(T+j)}\leq  M(\theta/M)\lambda^{j}=\theta \lambda^{j}\leq \theta,~j\in \mathbb{Z}_{+}.
$
This implies that 
$
\varphi_{x}(T+j)\in \mathcal{X},~j\in \mathbb{Z}_{+},
$
and 
$
\varphi_{x}\rightarrow 0_{n}
$
exponentially. If $\varphi_{x}(j)\in \mathcal{X}~\forall j\in \intcc{0;T-1}$, it follows that $\mathcal{V}(x)<\infty$, implying $x$ in an interior point of $\mathcal{D}_{0}^{\mathcal{X}}$, which yields a contradiction. Now, assume $\varphi_{x}(j)=y\in \mathbb{R}^{n}\setminus\mathcal{X}$ for some $j\in \intcc{0;T-1}$. Then, using the continuity of $f^{j}$, it follows that $\varphi_{x_{k}}(j)\rightarrow y$ as $k\rightarrow \infty$. This yields $\lim_{k\rightarrow \infty}\gamma(\varphi_{x_{k}}(j))= \infty$ and consequently $\mathcal{V}(x_{k})\rightarrow\infty$ as $k\rightarrow \infty$.
\end{proof}
\begin{corollary}
   $\mathcal{W}$ is continuous over $\mathbb{R}^{n}$. 
\end{corollary}

\section{Characterizing the Value Functions: Lyapunov and Zubov Equations}
\label{sec:Zubov}
In this section, we derive the Lyapunov and Zubov equations corresponding to the functions $\mathcal{V}$ and $\mathcal{W}$, respectively.

\begin{theorem}\label{Thm:LyapunovEquation}
For all $x\in \mathbb{R}^{n}$, $\mathcal{V}$ satisfies the maximal Lyapunov equation (w.r.t. to the function $v$)
\begin{align}\label{eq:LyapunovEqn}
v(x)%&= \sum_{k=0}^{\infty}\gamma({\varphi}_{x}(k))\alpha ({\varphi}_{x}(k))\\
%& =\gamma(x)\alpha(x)+\sum_{k=1}^{\infty}\gamma ({\varphi}_{x}(k))\alpha ({\varphi}_{x}(k))\\
= \gamma(x)\alpha(x)+v({f}(x)).
\end{align}
\end{theorem}
\begin{proof}
    Given $x\in \mathbb{R}^{n}$ and the definition of $\mathcal{V}$ in \eqref{eq:Lyapunov1}, we have 
       $ \mathcal{V}(x)= \sum_{k=0}^{\infty}\gamma({\varphi}_{x}(k))\alpha ({\varphi}_{x}(k))=\gamma(x)\alpha(x)+\sum_{k=1}^{\infty}\gamma ({\varphi}_{x}(k))\alpha ({\varphi}_{x}(k))=\gamma(x)\alpha(x)+\sum_{k=0}^{\infty}\gamma ({\varphi}_{f(x)}(k))\alpha ({\varphi}_{f(x)}(k))=\gamma(x)\alpha(x)+\mathcal{V}(f(x))$.
   Note that the above decomposition is valid even if $\mathcal{V}(x)$ is infinite.
\end{proof}

\begin{theorem}\label{thm:ZubovEqn1}
For all $x\in \mathbb{R}^{n}$, $\mathcal{W}$ satisfies the Zubov equation (w.r.t. to the function $w$)
\begin{equation}\label{eq:ZubovEqn1}
w(x)-w({f}(x))=\xi(x)(1-w(f(x))),
\end{equation}
where 
\begin{equation}\label{eq:xi}
\xi(x)\defas 1-\exp(-\gamma(x)\alpha(x)).
\end{equation}
\end{theorem}

\begin{proof}
Using Theorem \ref{Thm:LyapunovEquation}, we have, for any $x\in \mathbb{R}^{n}$, $\mathcal{W}(x)= 1-\exp(-\mathcal{V}(x))= 1-\exp(-\mathcal{V}({f}(x))-\gamma(x)\alpha(x))
        =1-\exp(-\gamma(x)\alpha(x))(1-\mathcal{W}({f}(x)))$,
implying 
$
     \mathcal{W}(x)-\mathcal{W}(f(x))=1-\mathcal{W}({f}(x))
     -\exp(-\gamma(x)\alpha(x))(1-\mathcal{W}({f}(x)))
     =\xi(x)(1-\mathcal{W}({f}(x)))$.
\end{proof}

\begin{theorem}\label{thm:ZubovEqn2}
If $w\colon D\subseteq\mathcal{X}\rightarrow \mathbb{R}$ satisfies equation \eqref{eq:ZubovEqn1} over $D$, then for  all $x\in D$, 
$w$ satisfies  the Zubov equation
\begin{equation}\label{eq:ZubovEqn2}
w(x)-w({f}(x))=\beta(x)(1-w(x)),
\end{equation}
where 
\begin{equation}\label{eq:beta}
\beta(x)\defas \exp(\gamma(x)\alpha(x))-1.
\end{equation}
\end{theorem}
\begin{proof}
    When $x\in D \subseteq \mathcal{X}$,  $\gamma(x)<\infty$, therefore, using equation \eqref{eq:ZubovEqn1}, 
    $
1-{w}({f}(x))=\exp(\gamma(x)\alpha(x))(1-w(x)).
    $
    Hence,
    $
       w(x)- w({f}(x))=\exp(\gamma(x)\alpha(x))(1-w(x))-1+w(x)= \exp(\gamma(x)\alpha(x))(1-w(x))-(1-w(x))= \beta(x)(1-w(x))$.  
\end{proof}
\subsection{Lyapunov and Zubov equations: Uniqueness results}
We have shown that the value functions $\mathcal{V}$ and $\mathcal{W}$ are solutions to the equations \eqref{eq:LyapunovEqn} and \eqref{eq:ZubovEqn1}, respectively. In the next section, we show that the solutions to these equations are unique with respect to functions the are continuous at the origin. We start with the following technical result:
\begin{lemma}\label{lem:boundedness}
    Assume that $w\colon \mathcal{D}_{0}^{\mathcal{X}}\rightarrow \mathbb{R}$ is continuous at the origin, with $w(0_{n})=0$ and satisfying equation \eqref{eq:ZubovEqn1} over $\mathcal{D}_{0}^{\mathcal{X}}$. Then $w(x)<1$ for all $x\in \mathcal{D}_{0}^{\mathcal{X}}$.
\end{lemma}
\begin{proof}
  Using Theorem \ref{thm:ZubovEqn2}, $w$ satisfies equation \eqref{eq:ZubovEqn2} over $\mathcal{D}_{0}^{\mathcal{X}}$. Assume that for some  $x\in \mathcal{D}_{0}^{\mathcal{X}}$, $w(x)\geq 1$. We have $\varphi_{x}(k)\in \mathcal{D}_{0}^{\mathcal{X}}~\forall k\in \mathbb{Z}_{+}$ and $\lim_{k\rightarrow \infty}\varphi_{x}(k)=0_{n}$. Using equation \eqref{eq:ZubovEqn2}, we have 
  $
w(x)-w(\varphi_{x}(1))=\beta(x)(1-w(x))\leq 0
  $
Hence, $w(\varphi_{x}(1))\geq w(x)\geq 1$, and by induction, we have $w(\varphi_{x}(k+1))\geq w(\varphi_{x}(k))\geq 1$ for all $k\in \mathbb{Z}_{+}$. Hence, $\lim_{k\rightarrow \infty}w(\varphi_{x}(k))\neq 0$, which yields a contradiction as $\lim_{k\rightarrow \infty}w(\varphi_{x}(k))$ must be zero due to the {\color{red}continuity} of $w$ at the origin, with ${\color{red}w(0_{n})=0}$ and $\lim_{k\rightarrow \infty}\varphi_{x}(k)=0_{n}$, and that completes the proof.
\end{proof}

\begin{theorem}\label{thm:LyapunovEqnUniqueness}
    Let $\mathbf{v}\colon \mathcal{D}_{0}^{\mathcal{X}}\rightarrow \mathbb{R}$ be a  function continuous at the origin and satisfying $\mathbf{v}(0_{n})=0$. Assume that $\mathbf{v}$ satisfies equation \eqref{eq:LyapunovEqn} over   
    $\mathcal{D}_{0}^{\mathcal{X}}$. Then
    $
\mathbf{v}(x)=\mathcal{V}(x)~\forall x\in \mathcal{D}_{0}^{\mathcal{X}}.
    $
\end{theorem}

\begin{proof}
    Let $x\in \mathcal{D}_{0}^{\mathcal{X}}$. As $\mathcal{D}_{0}^{\mathcal{X}}$ is invariant, we have 
    $\varphi_{x}(k)\in \mathcal{D}_{0}^{\mathcal{X}}~\forall k\in \mathbb{Z}_{+}.
    $
 Let $k\in \mathbb{N}$. We consequently have 
$ \mathbf{v}(x)-\mathbf{v}(\varphi_{x}(k))=\sum_{j=0}^{k-1}(\mathbf{v}(\varphi_{x}(j))-\mathbf{v}(\varphi_{x}(j+1)))
 =\sum_{j=0}^{k-1}\gamma(\varphi_{x}(j))\alpha(\varphi_{x}(j))$. As $\mathbf{v}$ is continuous at the origin with $\mathbf{v}(0_{n})=0$, $\varphi_{x}(k)\rightarrow 0_{n}$ as $k\rightarrow \infty$, and $\mathcal{V}(x)<\infty$ ($\sum_{j=0}^{k-1}\gamma(\varphi_{x}(j))\alpha(\varphi_{x}(j))$ converges), then taking the limit as $k\rightarrow \infty$ in the both sides of the above equation results in
    $
     \mathbf{v}(x)=\sum_{j=0}^{\infty}\gamma(\varphi_{x}(j))\alpha(\varphi_{x}(j))=\mathcal{V}(x).
    $
\end{proof}
\begin{theorem}
    Let $\mathbf{w}\colon \mathbb{R}^{n}\rightarrow \mathbb{R}$ be a bounded function, continuous at the origin, with $\mathbf{w}(0_{n})=0$, and  satisfying 
    equation \eqref{eq:ZubovEqn1} over $\mathbb{R}^{n}$.  
Then $\mathbf{w}=\mathcal{W}$.
\end{theorem}
\begin{proof}
   First note that the difference of two solutions, $\mathbf{w}_{1}$ and $\mathbf{w}_{2}$, to equation \eqref{eq:ZubovEqn1} satisfies, for $x\in \mathbb{R}^{n}$,
\begin{equation}\label{eq:difference1}
    \mathbf{w}_{1}(x)-\mathbf{w}_{2}(x)=(1-\xi(x))(\mathbf{w}_{1}(f(x))-\mathbf{w}_{2}(f(x))).
\end{equation}
When $x\in \mathbb{R}^{n}\setminus \mathcal{X}$, we have $\xi(x)=1$, implying $\mathbf{w}(x)=1=\mathcal{W}(x)$. Over  $ \mathcal{D}_{0}^{\mathcal{X}}$, $\mathbf{v}(\cdot)=-\ln(1-\mathbf{w}(\cdot))$ is well-defined due to Lemma \ref{lem:boundedness}, satisfying equation \eqref{eq:LyapunovEqn}. Then, using Theorem \ref{thm:LyapunovEqnUniqueness},  it follows that $\mathbf{v}(x)=\mathcal{V}(x)$, hence $\mathbf{w}(x)=\mathcal{W}(x)$ for all   $x\in \mathcal{D}_{0}^{\mathcal{X}}$.

Let $x\in \mathcal{X}$ be such that $\varphi_{x}(j)\in \mathbb{R}^{n}\setminus \mathcal{X}$ or $\varphi_{x}(j)\in \mathcal{D}_{0}^{\mathcal{X}}$  for some $j\in \mathbb{N}$ and $\varphi_{x}(k)\in \mathcal{X}~\forall k\in \intcc{0;j-1}$.  Then, we have, using equation \eqref{eq:difference1},  
$
\mathbf{w}(\varphi_{x}(j-1))-\mathcal{W}(\varphi_{x}(j-1))=(1-\xi(\varphi_{x}(j-1)))(\mathbf{w}(\varphi_{x}(j))-\mathcal{W}(\varphi_{x}(j)))=0$,
and by an inductive argument, we have 
$\mathbf{w}(\varphi_{x}(k))-\mathcal{W}(\varphi_{x}(k))=0~\forall k\in \intcc{0;j}$, implying $\mathbf{w}(x)=\mathcal{W}(x)$. 

 Now, let $x\in \mathcal{X}$ be such that $\varphi_{x}(k)\in \mathcal{X}$ for all $k\in \mathbb{Z}_{+}$ and $\lim_{k\rightarrow \infty}\varphi_{x}(k)\neq 0_{n}$. Obviously, $\varphi_{x}(k)\not\in \mathcal{D}_{0}^{\mathcal{X}}$ for all $k\in \mathbb{Z}_{+}$ 
then it follows, using the openness of $\mathcal{D}_{0}^{\mathcal{X}}$ that there exist $\theta>0$ such that $\norm{\varphi_{x}(k)}\geq \theta~\forall k\in \mathbb{Z}_{+}$. Assume $\mathbf{w}(x)\neq \mathcal{W}(x)$. 

Note the difference of two solutions, $\mathbf{w}_{1}$ and $\mathbf{w}_{2}$, to equation \eqref{eq:ZubovEqn1} over $\mathcal{X}$ (which are also solutions to \eqref{eq:ZubovEqn2} over $\mathcal{X}$ using Theorem \ref{thm:ZubovEqn2}) satisfies, for $x\in \mathcal{X}$,
\begin{equation}\label{eq:difference2}
    \mathbf{w}_{1}(f(x))-\mathbf{w}_{2}(f(x))=(1+\beta(x))(\mathbf{w}_{1}(x)-\mathbf{w}_{2}(x)).
\end{equation} 
It then follows that, for all $j\in \mathbb{N}$,   
$
\mathbf{w}(\varphi_{x}(j))-\mathcal{W}(\varphi_{x}(j))=\prod_{k=0}^{j-1}(1+\beta(\varphi_{x}(k)))(\mathbf{w}(x)-\mathcal{W}(x)).
$
For all $k \in \intcc{0;j-1}$, we have 
$
1+\beta(\varphi_{x}(k))=\exp(\gamma(\varphi_{x}(k))\alpha(\varphi_{x}(k))) \geq \exp(\underline{\gamma}\alpha_{m}\norm{\varphi_{x}(k)})\geq  \exp(\underline{\gamma}\alpha_{m}\theta^{2})$. Hence,
  $
\abs{\mathbf{w}(\varphi_{x}(j))-\mathcal{W}(\varphi_{x}(j))}\geq \exp(j\underline{\gamma}\alpha_{m}\theta^{2})\abs{\mathbf{w}(x)-\mathcal{W}(x)}.
  $
  As $\lim_{ j\rightarrow\infty}\exp(j\underline{\gamma}\alpha_{m}\theta^{2})=\infty$. It then follows that for each $M>0$, there exists $y\in \mathcal{X}$ (which corresponds to a point in the trajectory $\varphi_{x}$) such that  $\mathbf{w}(y)>M+\mathcal{W}(y)\geq M-1$ or $\mathbf{w}(y)<-M+\mathcal{W}(y)\leq -M+1=-(M-1)$, which is equivalent to $\abs{w(y)}> M-1$. Hence, $\mathbf{w}$ is unbounded over $\mathcal{X}$, a contradiction.
\end{proof}
\section{Physics-Informed neural solution}
\label{sec:PINN}

 Let $W_N(\cdot ; \theta)\colon \mathbb{R}^{n}\rightarrow \mathbb{R}$ be a fully-connected feedforward neural network with parameters denoted by $\theta$. We now train $W_N$ to solve Zubov's equation \eqref{eq:ZubovEqn1} on a compact set $\mathbb{X} \subset \mathbb{R}^n$, with $\mathcal{X} \cap \mathbb{X}\neq \emptyset$, by minimizing the following loss function:
\begin{equation}
\begin{split}
 \label{eq:lossV}
     \text{Loss}(\theta)=& 
     \frac{1}{N_c}\sum_{i=1}^{N_c} (W_N(x_{i};\theta)-W_N({f}(x_{i});\theta)-\\
    &\xi(x_{i})(1-W_N(f(x_{i});\theta)))^2  \\
    % \\& + \lambda_b \frac{1}{N_b} \sum_{i=1}^{N_b}(W_N(y_i;\theta)-B(y_i))^2\\
    & + \lambda_d \frac{1}{N_d} \sum_{i=1}^{N_d}(W_N(z_i;\theta)-\hat W(z_i))^2, 
    \end{split}
\end{equation}
 where $\lambda_b > 0$ and $\lambda_d > 0$ are user-defined weighting parameters. Here, the points $\{x_i\}_{i=1}^{N_c} \subset \mathbb{X}$ are the (interior) collocation points to compute the residual error of \eqref{eq:ZubovEqn1}. 
 % $B(y_i)$ denotes the boundary conditions of the Zubov equation, with the set $\{y_i\}_{i=1}^{N_b} \subset \partial \mathbb{X}$ consists of boundary points. Specifically, we only have  here. Moreover, if 
 To ensure an accurate neural network solution, we add the second term (refer to the data term) to guide the training, using the (approximate) ground truth values for $w$ at a set of points $\{z_i\}_{i=1}^{N_d} \subset \mathbb{X}$.
 
For each $z_i$, given a fixed (sufficiently large)  $M\in \mathbb{N}$, we compute $\mathcal{V}(z_i) \approx \sum_{k=0}^{M}\gamma({\varphi}_{x}(k))\alpha ({\varphi}_{x}(k))$. Note that as along as $\mathcal{V}(z_i)$ is finite, it is within the safe region. Next, we choose a large enough constant $C_{max}$ such that $C_{max} > \mathcal{V}(z_i)$, for all sampled $z_i$ whose trajectory does not leave the safe region and converges towards to the origin. Considering the fact that $1 - \exp(-40) \approx 1$, we introduce a scaling factor $\mu = \frac{40}{C_{max}}$ and define $\hat{\mathcal{V}}=\mu \mathcal{V}$, which corresponds to the value function given by (\ref{eq:Lyapunov1}) with $\alpha$ replaced by $\mu\alpha$. We then let $\hat W(z_i) = 1- \exp ( - \hat{\mathcal{V}}(z_i))$. Consequently, the approximate ground truth $\hat W$ takes values from 0 to 1, and $\hat W(z_i) < 1$ when it is in the safe region. Additionally, we set $\hat W(z_i) = 1$ in the following two scenarios: 1) if the trajectory $\varphi_{z_i}(k)$ that enters the unsafe region which results in $\mathcal{V}(z_i) > C_{max}$; 2) $\varphi_{z_i}(k)$ diverges, which can be numerically checked by examining the values of the states, that is, $\norm{\varphi_{z_i}(k)} > C_X$, where $C_X$ is a user-defined threshold.
% i.e., the pairs $(z_i, \hat{W}(z_i))\}_{i=1}^{N_d}$.
We also include the condition $\hat W(0) = 0$ in the data loss.

% In the next subsection, we will discuss in detail how to use PMP, presented in Algorithm \ref{alg:pmp}, to generate data points for solving 

\section{Verified Safe ROAs}
\label{sec:NNVerification}

\subsection{Safe ROAs with quadratic Lyapunov functions}
\label{sec:verify_quadratic}

Herein, we adopt the approach in \cite{serry2024safe} to find ellipsoidal safe ROAs using quadratic Lyapunov functions. We start with local safe ROA by linearization, where we assume that $f$ is twice-continuously differentiable. Let $A=Df(0)$ (assume $A$ to be a Schur matrix) and rewrite $f$ as 
$f(x)=Ax+h(x),~x\in \mathbb{R}^{n},$
where $h(\cdot)=f(\cdot)-A(\cdot)$. Let $Q\in \mathcal{S}^{n}_{++}$ be given, and $P\in \mathcal{S}^{n}_{++}$ be  the solution to the discrete-time algebraic Lyapunov equation
 $
A^{\intercal}PA-P=-Q.
 $
 Define the quadratic Lyapunov function 
 $
{V_P}(x)=x^{\intercal}Px,$ and let   $V_P^+ (x) :=  V_P (f(x)) - V_P(x)$. Define the positive parameter
$d\defas \underline{\lambda}(Q)-\varepsilon>0$,  for some sufficiently small $\varepsilon>0$.  Let $\mathcal{B}\subseteq \mathcal{X}$ be a hyper-rectangle with vector radius $R_{\mathcal{B}}\in \mathbb{R}^{n}_{+}\setminus \{0_{n}\}$, i.e., $\mathcal{B}=\Hintcc{-R_{\mathcal{B}},R_{\mathcal{B}}}$ (such a hyper-rectangle exists due to the openness of $\mathcal{X}$ and the fact that $0\in \mathcal{X}$).  We can find a vector $\eta_{\mathcal{B}}\in \mathbb{R}_{+}^{n}$ (by bounding the Hessian of $f$ over $\mathcal{B}$, e.g, using  interval arithmetic) such that 
$
\abs{h(x)}\leq \frac{\norm{x}^{2}}{2}\eta_{\mathcal{B}},~x\in \mathcal{B}
$. 
Define $
c_{1}=\min\{a_{1},a_{2}\}$, where
$$
a_{1}\defas \frac{(-\beta+\sqrt{\beta^2+4\alpha d})^2}{(2\alpha)^{2}},
$$ {\color{red}$\alpha\defas\norm{P}\norm{\eta_{\mathcal{B}}}^{2}\frac{1}{4\underline{\lambda}(P)}$} and $\beta\defas \norm{|P^{\frac{1}{2}}|\eta_{\mathcal{B}}}\norm{P^{\frac{1}{2}}AP^{-\frac{1}{2}}}$, and  
$$
a_{2}\defas \min_{i\in \intcc{1;n}} \frac{ \mathcal{R}_{\mathcal{B},i}^{2}}{P^{-1}_{i,i}}.
$$
Then, it can be shown that  \cite{serry2025underapproximatingsafe}\footnote{This is an archived version of \cite{serry2024safe}, which corrects errors in the bounds used in the derivations of the ellipsoidal ROAs.} for all $x\in \mathbb{R}^{n}$,
$$
{V}_{P}(x)<c_{1} \Rightarrow x\in \mathcal{B}\subseteq \mathcal{X} \wedge V_{P}^{+}(x)\leq -\varepsilon \norm{x}^{2}.
$$ 
It then follows that:
\begin{prop} \label{prop:local_quad}
    The set $ \mathbf{V}_{c_1} := \{x\in \mathbb{R}^n:\,V_P(x)\le c_1\}$ is a safe  ROA of \eqref{eq:System}. 
\end{prop}

% \begin{proof}
%     If \eqref{eq:c1_P} holds, then $\mathbf{V}_{c_1}$ is a subset of the safe set $\mathcal{X}$. Moreover, by (\ref{eq:v_p_dot}), $V_P$ is strict Lyapunov function for (\ref{eq:System}) and $\mathbf{V}_{c_1}$ is invariant. It follows by standard Lyapunov theory that $\lim_{k\rightarrow \infty} \varphi_{x}(k)=0_{n}, \quad \forall x \in \mathbf{V}_{c_1}$ and  $\mathbf{V}_{c_1}$ is a safe ROA of \eqref{eq:System}.
% \end{proof}

% By taking advantage of the fact that the origin is not contained in the unsafe set, we can find a $c_0 \leq c1$, such that $\{x\in X:\,V_P(x)\le c_0\} \subseteq \{x\in X:\, g(x)\le 1\} $, by continuity.
Suppose that we have verified a safe ROA around the origin, $\mathbf{V}_{c_1}$, for some $c_1>0$.
We then can enlarge the safe ROA with the quadratic Lyapunov function by verifying the following inequality, for $x\in \mathbb{R}^{n}$: 
\begin{align}
 c_1 \le {V}_{P}(x) \le c_2 \Longrightarrow (V_{P}^+(x)  \le -\varepsilon) \wedge (g(x) < 1),\label{eq:dVP}
\end{align}
% \begin{align}
%  (c_1 &\le {V}_{P}(x) \le c_2) \wedge (x\in X) \Longrightarrow \notag \\
% & (V_{P}^+(x)  \le -\varepsilon) \wedge (g(x) \leq 1) \wedge (f(x) \in X),\label{eq:dVP}
% \end{align}
where $c_2>c_1$ is a positive constant.
Unless otherwise specified, $\varepsilon$ is some positive constant in the context.
% Then, we have a larger local region of attraction $\{x\in \mathbb{R}^n:\, V_P(x)\le c_2\}$. With the second condition  we know it a safe local ROA.

\begin{prop} \label{prop:enlarge_quad}
    Suppose that \eqref{eq:dVP} holds. Then, the set $ \mathbf{V}_{c_2} := \{x\in \mathbb{R}^n:\,V_P(x)\le c_2\}$ is a safe ROA of \eqref{eq:System}. 
\end{prop}

\begin{proof}
    Define $\mathbf{V}_{c_2/c_1} := \{x\in \mathbb{R}^n:\, c_1 \le V_P(x)\le c_2\}$. If \eqref{eq:dVP} is satisfied, then all solutions $\varphi_x(k)$ of (\ref{eq:System}) starting in $\mathbf{V}_{c_2/c_1}$ cannot leave $\mathbf{V}_{c_2/c_1}$ until entering  $\mathbf{V}_{c_1}$ in finite time. By  Proposition \ref{prop:local_quad}, if a solution starts in or enters $\mathbf{V}_{c_1}$, it remains in $\mathbf{V}_{c_1}$ 
    and converges to $0_{n}$ eventually. Moreover, we have $g(x) < 1$ on both $\mathbf{V}_{c_2/c_1}$ and $\mathbf{V}_{c_1}$. It follows that $\mathbf{V}_{c_2} = \mathbf{V}_{c_2/c_1} \cup \mathbf{V}_{c_1} \subseteq \mathcal{X}$. Therefore, $\mathbf{V}_{c_2}$ is a safe ROA of \eqref{eq:System}. 
\end{proof}

\begin{remark} \label{rem:c2_max}
Assume $\mathbb{X}=\Hintcc{\underline{\mathbf{x}},\overline{\mathbf{x}}}$, i.e., $\mathbb{X}$ is a hyper-rectangle. 
   An upper bound on the  values of $c$ such that the set $\{x\in \mathbb{R}^{n}:\, x^{\top}Px\leq c\}\subseteq \mathbb{X}$ is needed as the condition \eqref{eq:dVP} is practically verified  over $\mathbb{X}$, and not over $\mathbb{R}^{n}$.  % Let $c\in \mathbb{R}_{+}$, $i\in \intcc{1;n}$, and $\mathbf{e}_{i}\in \mathbb{R}^{n}$ be the \textit{i}th standard basis vector. Then, by applying the transformation $y=P^{1/2}x$, we have   
   % \begin{align*}
%\max_{x^{\top}Px\leq c}x_{i}&=\max_{x^{\top}Px\leq c}x^{\top}\mathbf{e}_{i}=\max_{y^{\top}y\leq c} (P^{-1/2}y)^{\top}\mathbf{e}_{i}\\
%&= \max_{\norm{y}^{2}\leq c} (y)^{\top}P^{-1/2}\mathbf{e}_{i}=\sqrt{c} \norm{P^{-1/2}\mathbf{e}_{i}}\\
%&=\sqrt{c}\sqrt{\mathbf{e}_{i}^{\top}P^{-1}\mathbf{e}_{i}}
%=\sqrt{cP^{-1}_{i,i}}.
%    \end{align*}
%Using the symmetry of the set $\{x\in \mathbb{R}^{n}| x^{\top}Px\leq c\}$, we also have 
%$$
%\min_{x^{\top}Px\leq c}x_{i}=-\sqrt{cP^{-1}_{i,i}}.
%$$
A necessary and sufficient condition for inclusion is then %then
%$$
%(\sqrt{cP^{-1}_{i,i}} \leq \overline{\mathbf{x}}_{i})\wedge (-\sqrt{cP^{-1}_{i,i}} \geq \underline{\mathbf{x}}_{i})~\forall i\in \intcc{1;n}
%$$
%or equivalently,
%$$
%\sqrt{cP^{-1}_{i,i}}\leq \min\{-\underline{\mathbf{x}}_{i},\overline{\mathbf{x}}_{i}\} ~\forall i\in \intcc{1;n}.
%$$
%From the above condition, the upper bound on $c$ that ensures inclusion is given as:
$
c \leq \min_{i\in \intcc{1;n}} (\min\{-\underline{\mathbf{x}}_{i},\overline{\mathbf{x}}_{i}\})^{2}/P^{-1}_{i,i}.
$
\end{remark}

\subsection{Enlarge the ROAs with the neural Lyapunov functions} 
\label{sec:enlarge_neural}
% Assume $\mathbf{V}_{c_2} \subseteq X$.
The safe ROA estimate can be further enlarged if we can find a neural network Lyapunov function $W_N(x)$ by minimizing \eqref{eq:lossV} and verifying the following inequalities:
\begin{align}
(W_N(x)\le w_1) \wedge (x\in \mathbb{X})  & \Longrightarrow V_P(x)\le c_2, \label{eq:c2_inclusion} \\
(V_P(x)\le c_2) \wedge (x\in \mathbb{X})  &  \Longrightarrow W_N(x)\le w_2, \label{eq:w2_inclusion} \\
(w_1\le W_N(x) \le w_2) \wedge (x\in \mathbb{X}) & \Longrightarrow \notag \\ 
(W_N^+(x)  \le -\varepsilon) \wedge (g(x) & < 1) \wedge (f(x) \in \mathbb{X})\label{eq:dW}.
\end{align}
where $\varepsilon>0$, $w_2>w_1>0$, and $W_N^+(x) := W_N (f(x)) - W_N(x)$. 
% If (\ref{eq:dW}) and (\ref{eq:WP}) hold and the set $\mathcal{W}_{c_{V2}}=\{x\in X:\, W_N(x)\le c_{V2}\}$ does not intersect with the boundary of $X$, then $\mathcal{W}_{c_{V2}}$ is a safe ROA for (\ref{eq:System}).

\begin{prop}
    % \textcolor{blue}{Assume $\mathbf{V}_{c_2} \subseteq X$.} 
    Suppose that \eqref{eq:c2_inclusion}-\eqref{eq:dW} and the conditions in Proposition \ref{prop:enlarge_quad} hold. Then the set $\textbf{W}_{w_2} := \{x\in \mathbb{X}:\,W_N(x)\le w_2\}$ is a safe  ROA of \eqref{eq:System}. 
\end{prop}

\begin{proof}
    Define $\textbf{W}_{w_1} := \{x\in \mathbb{X}:\,W_N(x)\le w_1\}$ and  $\textbf{W}_{w_2/w_1} := \{x\in \mathbb{X}:\, w_1 < W_N(x)\le w_2\}$, where  $\textbf{W}_{w_2}=\textbf{W}_{w_1}\cup \textbf{W}_{w_2/w_1}$. 
    
    By conditions \eqref{eq:c2_inclusion} and \eqref{eq:w2_inclusion}, and proposition \eqref{prop:enlarge_quad},
    $\textbf{W}_{w_1}\subseteq \mathbf{V}_{c_{2}}\subseteq \textbf{W}_{w_2}\subseteq \mathcal{X}$.  Moreover, proposition \eqref{prop:enlarge_quad} implies that 
    \begin{equation}
    \label{eq:StartingFromWw1}
    x\in \mathbf{W}_{w_{1}}\Rightarrow \varphi_{x}(k)\in \mathbf{V}_{c_{2}}\subseteq \textbf{W}_{w_2},~ k\in \mathbb{Z}_{+},
    \end{equation}
    and $\lim_{k\rightarrow \infty}\varphi_{x}(k)=0_{n}$. 
    
    Let $x \in \mathbf{W}_{w_2/w_1}$.  Using condition \eqref{eq:dW}, we have $f(x) \in \mathbb{X}$ and $W_{N}(f(x))\leq w_{2}-\varepsilon\leq w_{2}$, i.e.,  $f(x)\in \textbf{W}_{w_2}$. This and  equation \eqref{eq:StartingFromWw1} imply the invariance of  $\mathbf{W}_{w_{2}}$.  
    
    It remains to show that, for $x \in \textbf{W}_{w_2/w_1}$,  there exists $N\in \mathbb{Z}_{+}$ such that $\varphi_{x}(N)\in \mathbf{W}_{w_{1}}$.    By contradiction, assume that $\varphi_{x}(k) \in \textbf{W}_{w_2/w_1}$ for all $k\in \mathbb{Z}_{+}$, then $w_1 < W_N(\varphi_{x}(k)) \le w_2$ for all $k\in \mathbb{Z}_{+}$. However, by inductively using   condition \eqref{eq:dW},  $W_N(\varphi_{x_0}(k))\leq w_{2} - k \epsilon,~k\in \mathbb{Z}_{+}$, which implies that for an integer $k>(w_{2}-w_{1})/\varepsilon$, $W_N(\varphi_{x_0}(k))< w_{1}$, a contradiction.  
\end{proof}

% \vspace{-1em}
\section{Numerical Examples}
\label{sec:NumericalExamples}
In this section, we present a set of numerical examples to illustrate the effectiveness of the proposed method. The training of the neural network Lyapunov functions is carried out using LyZNet~\cite{liu2024lyznet}, a Python toolbox for learning and verifying Lyapunov functions for nonlinear systems. For this work, we extended LyZNet to handle state-constrained, discrete-time nonlinear systems. For the verification part, the verification for the quadratic Lyapunov function is conducted with dReal~\cite{gao2013dreal} is conducted within the LyZNet framework as well, while the verification of the learned neural network Lyapunov function is with $\alpha,\!\beta$-CROWN \cite{zhang2018efficient, xu2020automatic, xu2021fast, wang2021beta, zhou2024scalable, shi2024genbab}, a GPU-accelerated neural network verifier. The verification with dReal was run on a 2x Intel Xeon 8480+ 2.0 GHz CPU with 24 cores, while GPU-required training and verification with $\alpha,\!\beta$-CROWN were performed on an NVIDIA Hopper H100 GPU. The code for the training and verification with dReal can be found at
\url{https://github.com/RuikunZhou/Safe_ROA_dt_LyZNet}, while the verification with $\alpha,\!\beta$-CROWN is available at \url{https://github.com/RuikunZhou/Safe-ROA-dt-system}.

% \vspace{-1em}
\subsection{Reversed Van der Pol Oscillator}

Consider a discrete-version of the reversed Van der Pol oscillator, which is the Euler discretization of the corresponding continuous-time system:
$$
f(x_{k})=    \begin{pmatrix}
    x_{1,k}-\Delta_{t}  x_{2,k}\\
       x_{2,k}+ \Delta_{t} (x_{1,k} +(x_{1,k}^2 - 1)x_{2,k} ))
    \end{pmatrix},
$$
where the step size $\Delta_{t} = 0.1$. We aim to estimate the safe DOA,  $\mathcal{D}_{0}^{\mathcal{X}}$, where $\mathcal{X}
=\mathbb{R}^2\setminus ([1~1]^{\top}+1/4\mathbb{B}_2)
%=\{x\in\mathbb{R}^2:\,(x_1-1)^2+(x_2-1)^2>1/16\}
$. 
In other words, there is an obstacle of radius $1/4$ centered at $[1~1]^{\top}$. The set on which training and verification take place is defined as 
$
\mathbb{X}=\Hintcc{[-2.5~-3.5]^{\intercal},[2.5~3.5]^{\intercal}}
$.
% $ \textcolor{blue}{[confirm notation here.]}
For the neural network training, we use a 2-hidden layer feedforward neural network with 30 neurons in each hidden layer and set $\alpha(x)=\Delta_{t}\norm{x}^{2}$, $g_{\mathcal{X}}(x_k) = 1 + 1/4 - ((x_{1,k} - 1)^2 + (x_{2,k} - 1)^2) / 0.25$. The learned neural network Lyapunov function and its corresponding safe ROA can be found in Fig.~\ref{fig:van_der_pol}. For the quadratic Lyapunov function, we choose $Q = I$ which results in $V_p = 16.3896 x_{1,k}^2 + 11.4027 x_{2,k}^2 - 11.1403 x_{1,k} x_{2,k}$, as demonstrated in Section~\ref{sec:verify_quadratic}. It is worth mentioning that we are able to compute an upper bound of the level set of the quadratic Lyapunov function within $\mathcal{X}$ using the method proposed in Remark~\ref{rem:c2_max} and then conduct bisection to determine $c_2$ using the verification tools. In this case, the upper bound is $85.43$. Consequently, we compute $c_1 = 1.78$ as illustrated in Section~\ref{sec:verify_quadratic}, and $c_2 = 11.14$ is determined using LyZNet. Then, we verify the conditions in Section~\ref{sec:enlarge_neural} for the neural network Lyapunov function. The time for verification with the two different verification tools for the neural network Lyapunov function is included in Table \ref{tab:verify_neural}.

\begin{figure} [h!t]
    \centering
    \includegraphics[width=\linewidth]{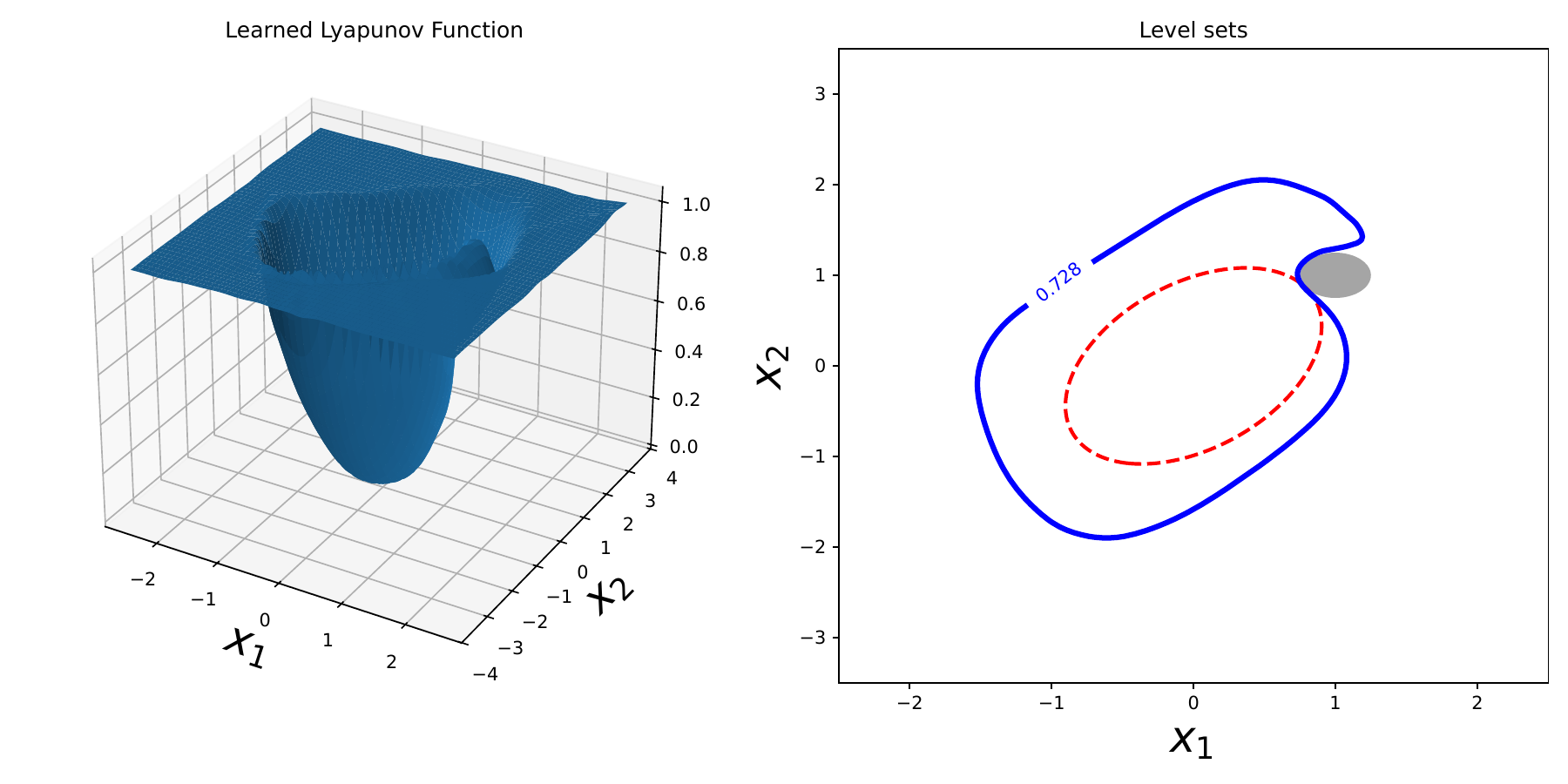}
    \caption{Neural Lyapunov function and the corresponding safe ROA (blue curve on the right) for the reversed Van der Pol oscillator, where the red dashed line represents the safe ROA obtained by the quadratic Lyapunov function with $c_2 = 11.14$.}
    \label{fig:van_der_pol}
\end{figure}

% \begin{table}[ht!]
%     \centering
%     \begin{tabular}{c|c|c}
%         \hline
%         System & dReal (s) & Crown (s) \\
%         \hline
%         \hline 
%         Van der Pol & \textbf{0.72} & 465.61 \\
%         \hline
%         Two-Machine & \textbf{0.71} & 378.82 \\
%         \hline
%         4D power system & 285,694.28 & \textbf{18,338.80} \\
%         \hline
%     \end{tabular}
%     \caption{Comparisons of the verification time for the quadratic Lyapunov function in Section~\ref{sec:verify_quadratic}}
%     \label{tab:verify_quadratic}
% \end{table}

\begin{table}[ht!]
    \centering
    \begin{tabular}{c|c|c}
        \hline
        System & dReal (s) & $\alpha,\!\beta$-CROWN (s) \\
        \hline
        \hline 
        Van der Pol & 401.69 & \textbf{358.05}\\
        \hline
        Two-Machine & 3490.81 & \textbf{295.06} \\
        \hline
        % 4D power system & -- &  \\
        % \hline
    \end{tabular}
    \caption{Comparisons of the verification time for the neural network Lyapunov function in Section~\ref{sec:enlarge_neural}}
    \label{tab:verify_neural}
\end{table}

\subsection{Two-Machine Power System}

In this subsection, we consider a discrete version of the two-dimensional two-machine power system studied in \cite{vannelli1985maximal,willems1968improved}. Same as the first example, the discrete version is obtained through Euler discretization of its continuous-time version, which yields
%\begin{equation}\label{eq:TwoMachine}
$$
f(x_{k})=    \begin{pmatrix}
    x_{1,k}+\Delta_{t}  x_{2,k}\\
       x_{2,k}-\Delta_{t}( \frac{x_{2,k}}{2}+\sin(x_{1,k}+\frac{\pi}{3})-\sin(\frac{\pi}{3}))
    \end{pmatrix},
$$
%\end{equation}
where $\Delta_{t} = 0.1$. Here, the safe set is set to be 
$
\mathcal{X}=\mathbb{R}^2\setminus ( ([0.25~0.25]^{\top}+(1/8)\mathbb{B}_{2})\cup ([0.25~ -0.25]^{\top}+(1/8)\mathbb{B}_{2}))
$. 
The set for training and verification is set to be $\mathbb{X}=\Hintcc{-[1~0.5]^{\intercal},[1~0.5]^{\intercal}}$.
We use the same $\alpha(x)$ as the first example and $g_{\mathcal{X}}(x_k) = \max \big( 1 + (1/8)^2 - \left( (x_{1,k} - 0.25)^2 + (x_{2,k} - 0.25)^2 \right),\ 1 + (1/8)^2  - \left( (x_{1,k} - 0.25)^2 + (x_{2,k} + 0.25)^2 \right) \big)$. With the same neural network structure and procedure as the ones in the previous example, both the quadratic Lyapunov function and neural network Lyapunov function can be obtained and verified. In this case, $V_P = 21.9377 x_{1,k}^2 + 33.6321 x_{2,k}^2 + 21.6816 x_{1,k} x_{2,k}$. In this example, we have $c_1 = 0.21$ and $c_2 = 0.86$. The safe ROAs and the learned neural network Lyapunov function are illustrated in Fig.~\ref{fig:two_machine}, while the verification time is reported in Table \ref{tab:verify_neural}. It is clear that in both two examples, the neural network Lyapunov functions obtained with the proposed method yield larger safe ROA estimates than the quadratic Lyapunov functions. 
% \textcolor{blue}{Due to space limit, } 
Additionally, $\alpha,\!\beta$-CROWN outperforms dReal in efficiency, even though it needs to be reinitialized in each iteration during bisection, which takes most of the reported time.

\begin{figure}[h!t]
    \centering
    \includegraphics[width=\linewidth]{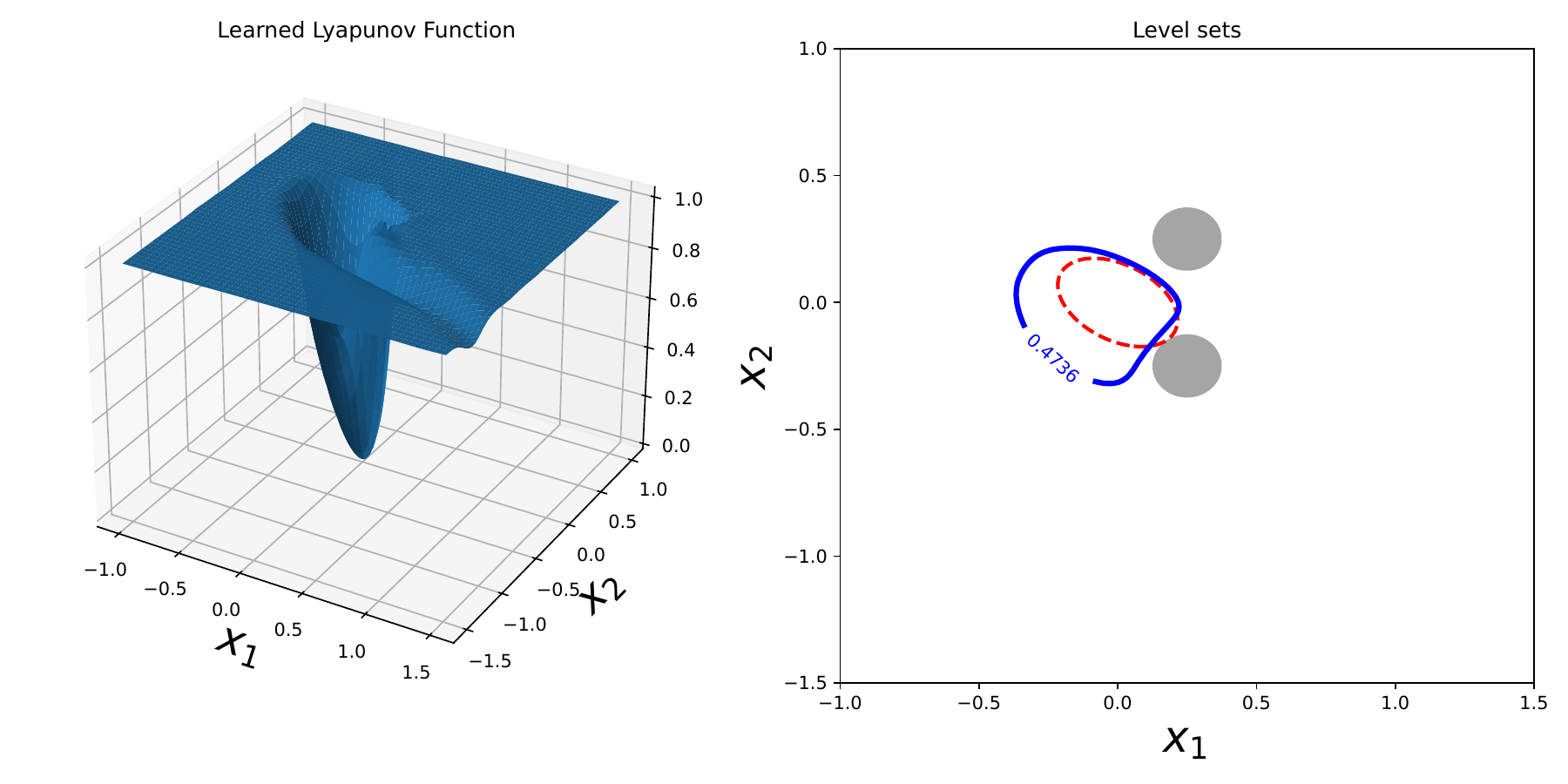}
    \caption{Neural Lyapunov function and the corresponding safe ROA (blue curve on the right) for the Two-Machine power system, where the red dashed line represents the safe ROA obtained by the quadratic Lyapunov function with $c_2 = 0.86$.}
    \label{fig:two_machine}
\end{figure}

\subsection{4-dimensional Power System}

We consider a 4-dimensional two generator bus power system in \cite[Chapter 5]{chiang2011direct}, as follows.
\[
f(x_{k})=  
\begin{pmatrix}
   x_{1,k} + \Delta_{t} x_{2,k} \\
   x_{2,k} + \Delta_{t} R_{2,k} \\
   x_{3,k} + \Delta_{t} x_{4,k} \\
   x_{4,k} + \Delta_{t} R_{4,k}
\end{pmatrix},
\]
where $R_{2,k} = \big( -\alpha_1 \sin(x_{1,k}) - \beta_1 \sin(x_{1,k} - x_{3,k}) - d_1 x_{2,k} \big)$ and $R_{4,k} = \big( -\alpha_2 \sin(x_{3,k}) - \beta_2 \sin(x_{3,k} - x_{1,k}) - d_2 x_{4,k} \big)$. The parameters are defined as follows: \( \alpha_1 = \alpha_2 = 1 \), \( \beta_1 = \beta_2 = 0.5 \), \( d_1 = 0.4 \), \( d_2 = 0.5 \), and the time step \( \Delta_{t} = 0.05 \). In this case, $\mathbb{X}=[-3.5~3.5]^4$, i.e., in each dimension, the range of the state is from $-3.5$ to $3.5$. Similar to the previous two examples, we learn a neural network Lyapunov function using a neural network with 2-hidden layers and 50 neurons each, while the quadratic Lyapunov function is also solved with $Q = I$. The expression of $V_P$ is omitted, but the constants are computed as $c_1 = 1.34$ and $c_2 = 140.625$.
For this high-dimensional system, obstacles are not included, i.e., $\mathcal{X}=\mathbb{R}^4$, as we mainly aim to illustrate the efficacy of the proposed method with formal guarantees provided by the more efficient verifier, $\alpha,\!\beta$-CROWN. Due to the limited scalability of dReal, the learned neural network Lyapunov function could not be verified with LyZNet within a 7-day timeout. The verified result using $\alpha,\!\beta$-CROWN is in Fig.~\ref{fig:power}. 
% \textcolor{blue}{[Don't have enough space to include all the details for $V_P$]}.

\begin{figure}[hbt]
    \centering
    \includegraphics[width=1\linewidth]{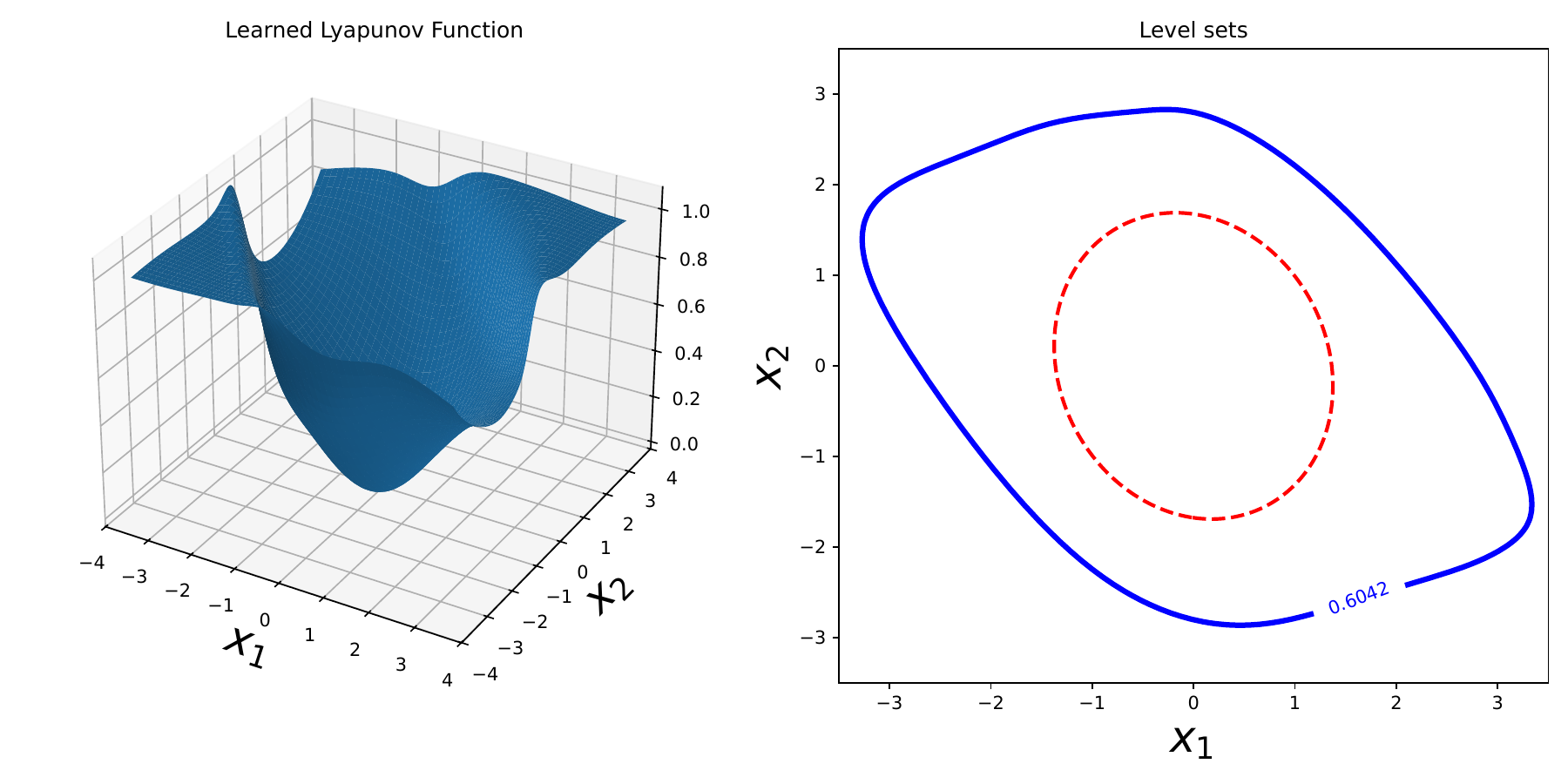}
    \caption{Neural Lyapunov function and the corresponding safe ROA (blue curve on the right) for the 4D power system, where the red dashed line represents the safe ROA obtained by the quadratic Lyapunov function with $c_2 = 140.625$.}
    \label{fig:power}
\end{figure}

\section{Conclusion}
\label{sec:Conclusion}

In this paper, we proposed a novel method for computing  verified safe ROAs for discrete-time systems by leveraging  neural networks approximate solutions to newly derived Zubov-type equations, and then obtaining certifiable ROAs (given as sublevel sets of the neural network solutions) using  verification tools. The verification framework relies on  initially obtaining  ellipsoidal ROAs and then enlarging them using neural network-based solutions. Our proposed framework was validated through three numerical examples,  illustrating its effectiveness. Future work includes the application of the proposed approach for handling disturbances to estimate a robust safe DOA, and the extension to the controlled case for computing safe null-controllability regions.

%\textcolor{blue}{some future work?}

%\bibliography{References} 
\bibliographystyle{ieeetr}

\end{document}